\documentclass[letterpaper,twocolumn,10pt]{article}

\usepackage[latin1]{inputenc}
\usepackage[english]{babel}
\usepackage{scs}
\usepackage{graphicx}
\usepackage{xspace}
\usepackage{amssymb,amsmath}
\usepackage{ amsthm}
\usepackage{hyperref}
\usepackage{gastex}

\newtheorem{theorem}{Theorem}
\newtheorem{definition}{Definition}
\newtheorem{corollary}{Corollary}

\newcommand{\ignore}[1]{}

\def    \ctl        {\mbox{\textsc{CTL }\xspace}}
\def    \CTL        {\mbox{\textsc{CTL }\xspace}}

\def    \U          {\mathcal{U}}
\def    \G          {\mathcal{G}}
\def    \F          {\mathcal{F}}

\def    \M          {{\cal M}}
\def    \X          {\mathcal{X}}
\def    \K          {\mathcal{K}}

\newcommand{\VHSM}{SHSM}
\newcommand{\HSM}{HSM}
\newcommand{\CHSM}{restricted \VHSM}
\newcommand{\boxhsm}{box}   
\newcommand{\boxes}{boxes}

\newcommand{\nnode}{node} 
\newcommand{\nnodes}{nodes} 
\newcommand{\expand} {\mathit{expn}}
\newcommand{\OUT} {\mbox{\sc out}}
\newcommand{\prop}{\mbox{{\sc true}}}
\newcommand{\vertex}{\mbox{vertex}}   
\newcommand{\vertices}{\mbox{vertices}}
\newcommand{\iin}{in}
\newcommand{\tuple}[1]{\langle #1 \rangle}
\newcommand{\word}{well-formed sequence}

\begin{document}

%
%
\title{Graded \CTL\ Model Checking for Test Generation}
\author{
Margherita Napoli and Mimmo Parente\\
Dip.to di Informatica ed Applicazioni\\
Universit\`a di Salerno, Italy\\
\href{mailto:napoli@unisa.it}{napoli@unisa.it}
\hspace{1truecm}
\href{mailto:parente@unisa.it}{parente@unisa.it}
}

\maketitle

\keywords{
Model Checking, Test Generation, Graded Temporal Logics, Hierarchical Finite State Machines.}

\begin{abstract}
\noindent
Recently there has been a great attention from the scientific community towards the use
of the model-checking technique as a tool for {\em test generation} in the simulation field.
This paper aims to provide a useful mean
to get more insights along these lines. By applying recent results
in the field of {\em graded} temporal logics, we present a new efficient model-checking algorithm
for Hierarchical Finite State Machines (\HSM), a well established symbolism long and widely
used for representing hierarchical models of discrete systems.
Performing model-checking against specifications expressed using graded temporal logics
has the peculiarity of returning more counterexamples within a {\em unique run}.
We think that this can greatly improve the efficacy
of automatically getting test cases.
In particular we ve\-ri\-fy two different models of  HSM
against branching time temporal properties.
\end{abstract}


%
%
\section{Introduction}\label{intro}
The {\em model-checking}  is a widely used technique to verify correctness of hardware and software systems.
A model checker explores the state space of a model of a given system to determine whether a given specification is satisfied.
Usually such specifications are expressed by means of formulas in a temporal logic, such as the Computational Temporal
Logics CTL, \cite{CE82}.
A very useful feature to fix the possible errors in the model is that when the model checker detects that the specification is violated then it returns a counterexample. In  the last years this feature  has also been exploited in
the simulation framework.
In fact, it is nowadays a well-established fact that formal (both software and hardware) analysis
is a valid complementary technique to simulation and testing (see e.g.,\cite{DHRPV07}).
On one side, the model checking approach, \cite{CGP99}, allows a full verification of system
components to be free of errors,
but its use is limited to small and medium sized models, due to the so-called state explosion
problem.
On the other hand the testing and simulation approaches \cite{PY} are usually applied to larger systems:
they  check the presence of errors in the system behavior through the observation of a chosen set of controlled executions.
Shortly, the efficacy of testing relies on the creation of test benches and that
of model-checking on the ability of formally defining the properties to be verified,
through temporal logic formulas.
More explicitly, the complementarity of the two techniques lies in the fact that the
counterexamples generated by a model-checker can be interpreted as test cases.
A good choice of the test suite is  the key for successful deductions of faults in
simulation processes. It is now more than a decade that model-checking is used for this purpose,
see \cite{FWA09, WASF07, A95, AB99, ABM98,GH99}.
In this context, a high level abstraction of the System Under Test (SUT),
is necessary.
Such abstraction should be simple and easy to model check,
but precise enough to serve as a basis for the generation of test cases.
This approach can be usefully adopted also in the DEVS modeling and simulation framework, \cite{Z76}.

However not surprisingly, the most challenging problem is the performance and two issues are crucial:
the choice of an efficient tool to generate the test suite and the choice of a suitable  abstract
model to check.

For the first issue, we propose the use of graded temporal logic specifications. In fact
standard model-checking tools generate only one counterexample for each
run and the check stage (of the model against a specification) is often expensive, in terms
of time resources.
We claim that it is highly desirable to get more meaningful counterexamples with a unique
run of the model checker.
For the second issue we propose the use of HSM as an abstract model of a DEVS
modeling the SUT, which preserves the
hierarchical structure while abstracting the continuous variables.
Thus we focus on how to generate simulation scenarios for
DEVS by providing a tool  which automatically generates multiple
counter-examples in an unique run, using hierarchical state machines as abstract model.
The sequence of events of each counterexample will
then be used to create a timed test trace for DEVS simulation.
In Figure~\ref{esempioAstrazione} a small example of our idea is shown
(the states labeled {\em Try1} and {\em Try2} are states on a higher hierarchy level
standing for the graph $M_1$).
Suppose we want to check whether the (timed) model in the figure satisfies the specification
(clearly false) stating that if a {\em Fail} occurs in the first attempt ({\em Try1}) of sending a message,
then an {\em Abort} event is eventually reached.
We can model-check an (untimed) over-approximation of the model (shown on the left)
obtaining the error trace {\em Start, Try1.(Send, Wait, Timeout, Fail), Try2.(Send, Wait, Ack), Success}.
This trace lets us concentrate on the portion of the model with a
potential error and can guide the simulation process to detect the error in the timed model.
\begin{figure*}[t]
\begin{picture}(70,50)(-40,-10)

\node[ExtNL=y, NLdist= 1, NLangle=165,Nw=73.0,Nh=20.0,Nmr=3.0](M2)(0,25){$M_2$}
\node[Nw=10,Nh=6,Nmr=3,Nmarks=i](n1)(-29,25){Start}
\node[Nw=10,Nh=6,Nmr=0](n2)(-13,25){Try$1$}
\node[Nw=10,Nh=6,Nmr=0](n3)(9,25){Try$2$}
\node[Nw=13,Nh=6,Nmr=3](n4)(27,30){Success}
\node[Nw=11,Nh=6,Nmr=3](n5)(27,20){Abort}
\drawedge(n1,n2){}
\drawedge(n2,n3){}
\drawedge(n3,n4){}
\drawedge(n3,n5){}
\drawedge[curvedepth=5](n2,n4){}

\drawline[dash={1.5}0](-18,22)(-28,10)
\drawline[dash={1.5}0](-8,22)(28,10)
\drawline[dash={1.5}0](4,22)(-28,10)
\drawline[dash={1.5}0](14,22)(28,10)

\node[ExtNL=y, NLdist= 1, NLangle=165,Nw=60.0,Nh=20.0,Nmr=3.0](M1)(0,0){$M_1$}
\node[Nw=10,Nh=6,Nmr=3,Nmarks=i](n6)(-23,0){Send}
\node[Nw=10,Nh=6,Nmr=3](n7)(-9,0){Wait}
\node[Nw=13,Nh=6,Nmr=3](n8)(9,-5){Timeout}
\node[Nw=10,Nh=6,Nmr=3,Nmarks=f](n9)(24,5){Ack}
\node[Nw=10,Nh=6,Nmr=3,Nmarks=f](n10)(24,-5){Fail}
\drawedge(n6,n7){}
\drawedge(n7,n8){}
\drawedge(n8,n10){}
\drawedge[curvedepth=3,ELside=r](n7,n9){}

\node[ExtNL=y, NLdist= 1, NLangle=165,Nw=73.0,Nh=20.0,Nmr=3.0](M2t)(95,25){$M_2$}
\node[Nw=10,Nh=6,Nmr=3,Nmarks=i](n11)(66,25){Start}
\node[Nw=10,Nh=6,Nmr=0](n12)(82,25){Try$1$}
\node[Nw=10,Nh=6,Nmr=0](n13)(104,25){Try$2$}
\node[Nw=13,Nh=6,Nmr=3](n14)(122,30){Success}
\node[Nw=11,Nh=6,Nmr=3](n15)(122,20){Abort}
\drawedge(n11,n12){}
\drawedge(n12,n13){{\small $t<10$}}
\drawedge(n13,n14){}
\drawedge(n13,n15){}
\drawedge[curvedepth=5](n12,n14){}

\drawline[dash={1.5}0](77,22)(67,10)
\drawline[dash={1.5}0](87,22)(123,10)
\drawline[dash={1.5}0](99,22)(67,10)
\drawline[dash={1.5}0](109,22)(123,10)

\node[ExtNL=y, NLdist= 1, NLangle=165,Nw=60.0,Nh=20.0,Nmr=3.0](M1t)(95,0){$M_1$}
\node[Nw=10,Nh=6,Nmr=3,Nmarks=i](n16)(72,0){Send}
\node[Nw=10,Nh=6,Nmr=3](n17)(86,0){Wait}
\node[Nw=13,Nh=6,Nmr=3](n18)(104,-5){Timeout}
\node[Nw=10,Nh=6,Nmr=3,Nmarks=f](n19)(119,5){Ack}
\node[Nw=10,Nh=6,Nmr=3,Nmarks=f](n110)(119,-5){Fail}
\drawedge(n16,n17){}
\drawedge(n17,n18){{\small $t>2$}}
\drawedge(n18,n110){}
\drawedge[curvedepth=3,ELside=r](n17,n19){}

\end{picture}
\caption{An over-approximation of a model (untimed on the left and timed on the right).}\label{esempioAstrazione}
\end{figure*}
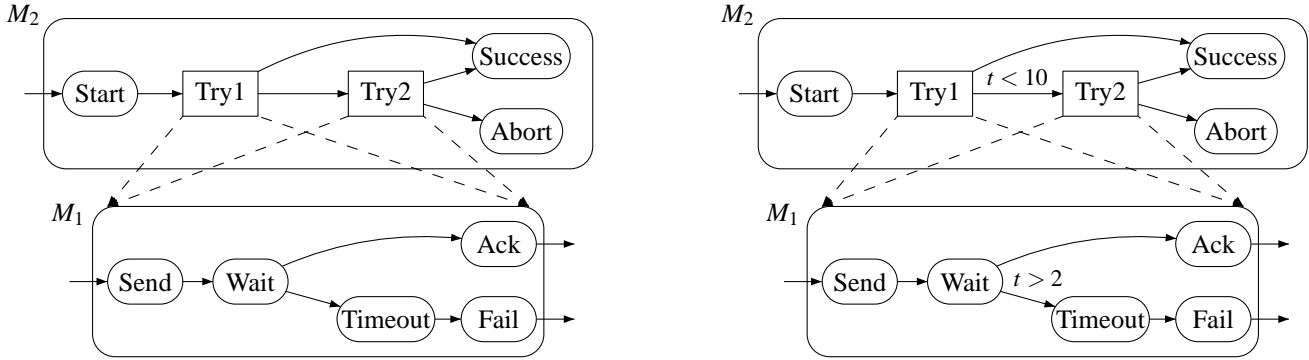
Let us now briefly detail the two notions of graded logics and HSM.
In order to get more counterexamples in a unique run we use
specifications expressed in {\em graded}-\CTL, recently introduced in \cite{FNP08}.
Graded-\CTL\ strictly extends classical \CTL\ with graded modalities: classical CTL can be used for reasoning
about the temporal behavior of systems con\-si\-de\-ring either {\em all} the possible futures
or {\em at least one} possible future, while graded-CTL uses graded extensions on
both exi\-stential and universal quantifiers.
With graded-\CTL\ formulas one can describe a constant number of future scenarios. For example,
one can express that in $k$ different cases it is possible that a waiting process never
obtains a requested resource, or that there are $k$ different ways for a system to reach a
{\em safe state} from a given state.

The notion of finite state machine with a hierarchical structure
has been used for many years for modelling discrete systems,
since the introduction of Statecharts, \cite{H87}, and is actually
applied into many fields as a specification formalism.
In particular, in the  model-checking framework, one of the most considered models is the Hierarchical State Machine (\HSM)
 (see e.g. \cite{AY01}).
A generalization of \HSM\ is introduced
in \cite{LNPP08}, as an exponentially more succinct model  where also higher level states, called {\em boxes},  are labeled with atomic propositions.
The intended meaning of such labeling is that when a box $b$
expands to a machine $M$, all the vertices of $M$ \emph{inherit} the
atomic propositions of $b$ (\emph{scope}), such that different
vertices expanding to $M$ can place $M$ into different scopes.
Such model is called a \emph{hierarchical state machine
with scope-dependent properties} (Scope-dependent Hierarchical State Machine, shortly \VHSM).

Our contribution aims in providing also strong theoretical evidence of
the soundness of our approach.
In particular we study the problem of verifying whether an \VHSM\ models a given graded-\ctl\ formula.
We first give an algorithm to solve the graded-\ctl\ model-checking  of an \HSM, and then
we extend it to model-check general \VHSM s.
We show that the problem has the same computational complexity as
 \ctl\ model checking, and we show how to solve it both for  \HSM\ and  \VHSM, with an extra factor in the exponent which is
logarithmic in the maximal grading constant occurring in the \CTL\ formula.
%
Let us stress that the experimental results for flat models reported in \cite{FMNPS10} shows that
 this  extra factor does not have real effects in the running time of the algorithms
(currently we are implementing also the algorithms presented here for
hierarchical structures and the initial tests are very promising).

The rest of the paper is organized as follows:
in Sections~\ref{sec:gradedDefinition}
and~\ref{sec:hierarDefinition}
we give basic definitions and known results
of graded-CTL, and of SHSM, respectively;
in Section~\ref{sec:Algo} we give the algorithm to model-check \VHSM\
against graded-CTL  specifications.
In Section~\ref{sec:Conclusions} we give our conclusions.

\section{Graded \CTL}\label{sec:gradedDefinition}
In this section we first recall the definitions of \ctl and then give that of
graded-\ctl,
see \cite{FNP08}.
The temporal logic \ctl \cite{CE82} is a branching-time logic in
which each temporal operator, expressing properties about a
possible future, has to be preceded either by an  existential or
by an universal path quantifier. So, in \ctl
one can express properties that have to be true either
\emph{immediately after now} ($\X$), or \emph{each time from now}
($\G$), or \emph{from now until something happens} ($\U$), and it
is possible to specify that each property must hold either in
\emph{some possible futures} ($E$) or in \emph{each possible
future} ($A$). Formally, given a finite set of \emph{atomic
propositions} $AP$, \ctl is the set of formulas
$\varphi$ defined as follows:
    $$\varphi := p\ |\ \neg\psi_1\ |\ \psi_1 \wedge \psi_2\ |\ E \X \psi_1\ |\ E \G \psi_1\ |\ E \psi_1 \U \psi_2$$
where $p \in AP$ is an atomic proposition and $\psi_1$ and
$\psi_2$ are \ctl formulas.
The semantics of a \ctl formula is defined with respect to a
\emph{Kripke Structure} by means of the classical relation
$\models$. As usual, a Kripke structure over a set of atomic
propositions $AP$, is a tuple  $\K = \langle S, s_{in}, R, L
\rangle$, where $S$ is a finite set of states, $s_{in} \in S$ is
the initial state, $R \subseteq S \times S$ is a transition
relation with the property that for each $s \in S$ there is $t \in
S$ such that $(s,t) \in R$, and $L : S \rightarrow 2^{AP}$ is a
labeling function.
A path in $\K$ is denoted by the sequence of states $\pi = \langle
s_0, s_1, \ldots s_n\rangle$ or by $\pi = \langle s_0, s_1, \ldots
\rangle$, if it is infinite. The length of a path, denoted by
$|\pi|$, is the number of states in the sequence, and $\pi[i]$
denotes the $i$-th state $s_i$.
Then, the relation $\models$ for a state $s \in S$ of $\K$ is
iteratively defined as follows:
\begin{itemize}
\item
$(\K, s) \models p \in AP$ iff $p \in L(s)$;
\item
$(\K, s) \models \neg\psi_1$ iff $\neg( (\K, s) \models \psi_1)$
(in short, $(\K, s) \not\models \psi_1$);
\item
$(\K, s) \models \psi_1 \wedge \psi_2$ iff $(\K, s) \models
\psi_1$ and $(\K, s) \models \psi_2$;
\item
$(\K, s) \models E \X \psi_1$ iff there exists $s' \in S$ such
that $(s, s') \in R$ and $(\K, s') \models \psi_1$ (the path
$\langle s, s' \rangle$ is called an \emph{evidence} of the formula
$\X \psi_1$);
\item
$(\K, s) \models E \G \psi_1$ iff there exists an infinite path
$\pi$ starting from $s$ (i.e., $\pi[0] = s$) such that for all $j
\geq 0$, $(\K, \pi[j]) \models \psi_1$ (the path $\pi$ is called
an \emph{evidence} of the formula $\G \psi_1$);
\item
$(\K, s) \models E \psi_1 \U \psi_2$ iff there exists a finite
path $\pi$ with length $|\pi| = r+1$ starting from $s$ such that
$(\K, \pi[r]) \models \psi_2$ and, for all $0 \leq j < r$, $(\K,
\pi[j]) \models \psi_1$ (the path $\pi$ is called an
\emph{evidence} of the formula $\psi_1 \U \psi_2$);
\end{itemize}
We say that a Kripke structure $\K = \langle S, s_{in}, R, L
\rangle$ \emph{models} a \ctl formula $\varphi$ iff $(\K, s_{in})
\models \varphi$.
Note that we have expressed the syntax of \ctl with one of the
possible minimal sets of operators. Other temporal operators as well as
the universal path quantifier $A$,  can be easily derived
from those.
\textbf{Graded-\ctl} extends the classical \ctl by adding graded modalities on the quantifier
operators. Graded modalities specify in how many possible futures  a given path property has to hold,
and thus generalize \ctl allowing to reason
about more than a given number of possible distinct future
behaviors. Let us first define the notion of {\em distinct}.
Let $\K = \langle S, s_{in}, R,
L \rangle$ be a Kripke structure. We say that two paths $\pi_1$
and $\pi_2$ on $\K$ are \emph{distinct} if there exists an index
$0 \leq i < \min\{ |\pi_1|, |\pi_2| \}$ such that $\pi_1[i] \neq
\pi_2[i]$. Observe that from this definition if a path is the
prefix of another path, then they are not distinct.
The \emph{graded existential path quantifier} $E^{>k}$, requires
the existence of $k+1$ pairwise distinct evidences of a path-formula. Given a
set of atomic proposition $AP$, the syntax of graded-\ctl is
defined as follows:
    $$\varphi := p\ |\ \neg \psi_1\ |\ \psi_1 \wedge \psi_2\ |\ E^{> k} \X \psi_1\ |\ E^{> k} \G \psi_1\ |\  E^{> k} \psi_1 \U \psi_2$$
where $p \in AP$, $k$ is a non-negative integer and $\psi_1$ and
$\psi_2$ are graded-\ctl formulas.
The semantics of graded-\ctl is still defined with respect to a
Kripke structure $\K = \langle S, s_{in}, R, L \rangle$ on the set
of atomic propositions $AP$. In particular, for formulas of the
form $p$, $\neg\psi_1$ and $\psi_1 \wedge \psi_2$ the semantics is the same as in the classical \ctl.
For the remaining formulas, the semantics is defined as follows:
\begin{itemize}
\item
$(\K, s) \models E^{> k} \theta$, with $k \geq 0$ and either $\theta = \X \psi_1$ or $\theta = \G \psi_1$ or $\theta = \psi_1 \U \psi_2$, iff there exist $k+1$ pairwise
distinct evidences of $\theta$ starting from $s$.
\end{itemize}
It is easy to observe that classical \ctl is a proper fragment of
graded-\ctl since the simple graded formula $E^{>1} \X p$ cannot be
expressed in \ctl, whereas any \ctl formula is also a graded-\ctl
formula  (note that $E^{>0} \theta$ is equivalent to $E \theta$).
We can also consider the graded extension of the universal
quantifier, $A^{\leq k}$, with the meaning that \emph{all the
paths starting from a node $s$, but at most $k$ pairwise distinct
paths, are evidences of a given path-formula}. The quantifier
$A^{\leq k}$ is the dual operator of $E^{>k}$ and can obviously be
re-written in terms of $\neg E^{>k}$. However, while $A^{\leq k}
\X \psi_1$ and $A^{\leq k} \G \psi_1$ can be easily re-written
respectively as $\neg E^{>k} \X \neg \psi_1$ and $\neg E^{>k} \F
\neg\psi_1$, the transformation of the formula $A^{\leq k} \psi_1
\U \psi_2$ with $k>0$ in terms of $\neg E^{>k}$ deserves more care
(see \cite{FNP08} for a detailed treatment).
\ignore{
In fact, we have that $A^{\leq k} \psi_1 \U \psi_2$
is equivalent to $\neg E^{> k} \neg(\psi_1 \U \psi_2)$ (note that this formula is
not a graded-\ctl formula because of the occurrence of the
innermost negation), that can be translated in graded-\ctl in the
following way:
\begin{eqnarray}
    A^{\leq k} \psi_1 \U \psi_2 &\Longleftrightarrow& \neg E^{>k} \G (\psi_1 \wedge \neg\psi_2)\ \wedge \neg E^{>k} (\psi_1 \wedge \neg\psi_2) \U (\neg\psi_1 \wedge \neg\psi_2) \wedge \label{for:ForallUntil} \\
    & & \bigwedge_{i=0}^{k-1} (\ \neg E^{>k-1-i} \G (\psi_1 \wedge \neg\psi_2)\ \vee\ \neg E^{>i} (\psi_1 \wedge \neg\psi_2) \U (\neg\psi_1 \wedge \neg\psi_2)\ ) \nonumber
\end{eqnarray}

In fact observe that a path not
satisfying $\psi_1 \U \psi_2$ is a path that satisfies either
$\theta_1 = \G (\psi_1 \wedge \neg\psi_2)$ or $\theta_2 = (\psi_1
\wedge \neg\psi_2) \U (\neg\psi_1 \wedge \neg\psi_2)$ (clearly,
the paths satisfying $\theta_1$ are all distinct from the paths
satisfying $\theta_2$). Therefore the formula $E^{>k} \neg(\psi_1
\U \psi_2)$ holds in $s$, if $k+1$ pairwise distinct paths stem
from this, each satisfying either $\theta_1$ or $\theta_2$.
}

The \textbf{graded-\ctl model-checking} is the problem of
verifying whether a Kripke structure $\K$ models a graded-\ctl
formula $\varphi$.
The complexity of the graded-\ctl model-checking problem
is linear with respect to the size of
the Kripke structure and to the size of the formula,
(this latter being the number of the temporal and the boolean
operators occurring in it). Let us remark that this complexity is
independent from the integers $k$ occurring in the formula.

\section{Scope-dependent  Hierarchical State Machines}\label{sec:hierarDefinition}
In this section we  formally define the Scope-dependent Hierarchical State Machines and
recall some known results.
The Scope-dependent Hierarchical State Machines are defined as follows.
\begin{definition}
A {\em Scope-dependent Hierarchical State Machine}
(\VHSM) over $AP$ is a tuple ${\cal M}= (M_1,M_2,\ldots,M_h)$, each
$M_i=(V_i,in_i,\OUT_i, \prop_i,\expand_i, E_i)$ is called {\em
machine} and consists of:
\begin{itemize}
\item a finite set of \vertices\ $V_i$,
an {\em initial} \vertex\ $in_i \in V_i$
and a set of {\em output} \vertices\ $\OUT_i \subseteq V_i$;
\item a labeling function $ \prop_i: V_i \longrightarrow 2^{AP}$ that
maps each \vertex\ with a set of atomic propositions;
\item an expansion mapping $\expand_i: V_i \longrightarrow \{0,1,\ldots,h\} $
such that $\expand_i(u) < i$, for each $u \in V_i$, and
$\expand_i(u)=0$, for each $u \in \{in_i \} \cup \OUT_i$;
\item a set of edges $E_i$ where each edge is
either a couple $(u, v)$, with $u, v \in V_i$ and
$\expand_i(u) =0$, or a triple $((u,z),v)$ with
$u,v\in V_i$, $\expand_i(u)=j, j>0$, and
$z \in \OUT_j$.
\end{itemize}
\end{definition}
\begin{figure*}[t]
\begin{center}
\begin{picture}(137,52)(-8,-52)
\node[NLangle=163.0,NLdist=35.0,Nw=74.0,Nh=17.0,Nmr=2.98](n0)(34.55,-16.19){ $M_3$}

\node[Nfill=y,fillgray=0.9,NLangle=-90.0,NLdist=6.0,Nw=15.0,Nmr=1.0](n39)(46.55,-16.19){
$\{p_3\}$} \nodelabel[NLangle=0.0](n39){$b_3^1$}

\node[Nfill=y,fillgray=1.0,Nw=2.2,Nh=2.2,Nmr=1.1](n41)(40.55,-16.19){}

\node[Nfill=y,fillgray=0.9,NLangle=0.0,Nw=15.0,Nmr=1.0](n56)(22.55,-16.19){
$b_3^0$} \nodelabel[NLangle=-90.0,NLdist=6.0](n56){$\emptyset$}

\node[Nfill=y,fillgray=1.0,Nw=2.2,Nh=2.2,Nmr=1.1](n57)(28.55,-16.19){}

\node[Nfill=y,fillgray=1.0,Nw=2.2,Nh=2.2,Nmr=1.1](n58)(16.65,-16.19){}

\node[Nfill=y,fillgray=1.0,NLangle=0.0,ilength=8.0,Nmarks=i,Nw=4.5,Nh=4.5,Nmr=2.25](n108)(4.55,-16.19){
$\iin_3$}
\nodelabel[NLangle=-90.0,NLdist=5.0](n108){$\emptyset$}

\drawedge(n108,n58){}

\node[Nfill=y,fillgray=1.0,Nw=2.2,Nh=2.2,Nmr=1.1](n76)(52.55,-16.19){}

\node[Nfill=y,fillgray=1.0,NLangle=0.0,flength=8.0,Nmarks=f,Nw=4.5,Nh=4.5,Nmr=2.25](n188)(64.55,-16.19){
$z_3$} \nodelabel[NLangle=-90.0,NLdist=5.0](n188){
$\{p_3,p_2,p_1\}$}

\drawedge(n57,n41){}

\drawedge(n76,n188){}

\drawloop[loopdiam=5.0,loopangle=50.0](n108){}

\drawloop[loopdiam=5.0,loopangle=50.0](n188){}

\node[Nfill=y,fillgray=0.8,NLangle=149.0,NLdist=20.5,Nw=44.0,Nh=17.0,Nmr=2.98](n230)(102.71,-40.32){
$M_1$}

\node[Nfill=y,fillgray=1.0,NLangle=0.0,ilength=8.0,Nmarks=i,Nw=4.5,Nh=4.5,Nmr=2.25](n232)(87.71,-40.32){
$\iin_1$}
\nodelabel[NLangle=-90.0,NLdist=5.0](n232){$\emptyset$}

\node[Nfill=y,fillgray=1.0,NLangle=0.0,flength=8.0,Nmarks=f,Nw=4.5,Nh=4.5,Nmr=2.25](n234)(117.71,-40.32){
$z_1$} \nodelabel[NLangle=-90.0,NLdist=5.0](n234){$\{p_1\}$}

\drawloop[loopdiam=5.0,loopangle=50.0](n232){}

\drawloop[loopdiam=5.0,loopangle=50.0](n234){}

\drawedge(n232,n234){}

\node[Nfill=y,fillgray=0.9,NLangle=163.0,NLdist=35.0,Nw=74.0,Nh=17.0,Nmr=2.89](n95)(34.55,-40.52){
$M_2$}

\node[Nfill=y,fillgray=0.8,NLangle=-90.0,NLdist=6.0,Nw=15.0,Nmr=1.0](n96)(46.55,-40.52){
$\{p_2\}$} \nodelabel[NLangle=0.0](n96){$b_2^1$}

\node[Nfill=y,fillgray=1.0,Nw=2.2,Nh=2.2,Nmr=1.1](n97)(40.55,-40.52){}

\node[Nfill=y,fillgray=0.8,NLangle=0.0,Nw=15.0,Nmr=1.0](n98)(22.55,-40.52){
$b_2^0$} \nodelabel[NLangle=-90.0,NLdist=6.0](n98){
$\emptyset$}

\node[Nfill=y,fillgray=1.0,Nw=2.2,Nh=2.2,Nmr=1.1](n99)(28.55,-40.52){}

\node[Nfill=y,fillgray=1.0,Nw=2.2,Nh=2.2,Nmr=1.1](n100)(16.55,-40.52){}

\node[Nfill=y,fillgray=1.0,NLangle=0.0,ilength=8.0,Nmarks=i,Nw=4.5,Nh=4.5,Nmr=2.25](n101)(4.55,-40.52){
$\iin_2$}
\nodelabel[NLangle=-90.0,NLdist=5.0](n101){$\emptyset$}

\drawedge(n101,n100){}

\node[Nfill=y,fillgray=1.0,Nw=2.2,Nh=2.2,Nmr=1.1](n102)(52.55,-40.52){}

\node[Nfill=y,fillgray=1.0,NLangle=0.0,flength=8.0,Nmarks=f,Nw=4.5,Nh=4.5,Nmr=2.25](n103)(64.55,-40.52){$z_2$}
\nodelabel[NLangle=-90.0,NLdist=5.0](n103){$\{p_2,p_1\}$}

\drawedge(n99,n97){}

\drawedge(n102,n103){}

\drawloop[loopdiam=5.0,loopangle=50.0](n101){}

\drawloop[loopdiam=5.0,loopangle=50.0](n103){}

\end{picture}
\end{center}

\caption{A simple \VHSM\ $\M$.  }\label{esempioCHSM}
\end{figure*}
In the rest of the paper we use $h$ as the number of machines of
an \VHSM\ $\M$ and $M_h$ is called {\em top-level} machine.
We assume that the sets of \vertices\ $V_i$ are pairwise disjoint.
The set of all \vertices\ of $\M$ is $V=\bigcup_{i=1}^{h} V_i$. The
mappings $\expand: V \longrightarrow \{0,1,\ldots,h\} $ and $\prop:
V \longrightarrow 2^{AP}$ extend the mappings $\expand_i$ and
$\prop_i$, respectively. If $\expand(u)=j>0$, the \vertex\ $u$
expands to the machine $M_j$ and is called {\em \boxhsm}. When
$\expand(u)=0$, $u$ is called a {\em \nnode}. Let us define the
closure $\expand^{+}: V \longrightarrow 2^{\{0,1,\ldots,h\}}$, as:
$h \in \expand^{+}(u)$ if either $h=\expand(u)$ or there exists
$u'\in V_{\expand(u)}$ such that $h \in \expand^{+}(u')$. We say
that a \vertex\ $u$ is an \emph{ancestor} of $v$ and $v$ is a
\emph{descendant} from $u$ if $v \in V_h$, for $h \in
\expand^{+}(u)$.

A vertex $v\in V_i$ is called a \emph{successor} of $u\in V_i$ if there is an edge
 $(u, v)\in E_i$, and it is called a \emph{z-successor} of $u$,
 for $z \in \OUT_{\expand(u)}$,  if $((u,z),v)\in E_i $.

An \HSM\ is an \VHSM\ such that  $\prop(b)= \emptyset$,
for any box $b$.

As an example of an \VHSM\ $\M$ see
Figure~\ref{esempioCHSM}, where $p_1, p_2, p_3$ are atomic
propositions labeling \nnodes\ and \boxes\ of $\M$, $\iin_i$ and $z_i$
are respectively entry \nnodes\ and exit \nnodes\ for $i=1,2,3$, and
$\expand(b^i_j)=j-1$ for $i=0,1$ and $j=2,3$.

\ignore{
\begin{definition}\label{restricted}
A restricted \VHSM\ ${\cal M}$ is an \VHSM\
where for all vertices $u,v$
such that $u$ is an ancestor of $v$ in $\M$ it holds:\\
\centerline{$\prop(u) \cap \prop(v)
= \emptyset .$}
\end{definition}

Such a restriction is quite natural and still allows us to succinctly represent
interesting systems. Note that
the \VHSM\ of Figure~\ref{esempioCHSM} is also restricted.
}


\noindent{\large\bf Semantics.}
The semantics of an \VHSM\ $\M$
is given by a {\em flat}  Kripke structure, denoted $\M^F$.

A sequence of \vertices\ $\alpha=u_1 \ldots u_m$, $1 \leq m$, is
called a {\em \word} if $u_{\ell+1} \in V_{\expand(u_\ell)}$, for
$\ell =1,\ldots,m-1$.
Moreover, $\alpha$ is also {\em complete} when $u_1 \in V_h$ and
$u_m$ is a \nnode.

A state of $\M^F$ is $\langle \alpha\rangle$ where $\alpha$ is
a complete \word\ of $\M$.
Note that the length of a complete \word\ is at most $h$, therefore
the number of states  of $\M^F$ is at most exponential in the number of
machines composing $\M$.
Transitions of $\M^F$ are obtained by using as templates the
edges of $\M$.
Figure~\ref{peggioDiNoi} shows the Kripke structure
which is equivalent to the \VHSM\ of Figure~\ref{esempioCHSM}.
We formally define $\M^F$ as follows.
Given an \VHSM\ $\M=(M_1,M_2,\ldots,M_h)$, it is immediate to
observe that the tuple $\M_j=(M_1,M_2,\ldots,M_j)$, $1 \leq j \leq
h$, is an \VHSM\ as well. Clearly, $\M_h=\M$. In the following, we
sketch how to compute recursively the flat Kripke structures
$\M_j^F$.

We start with $\M_1^F$ which is obtained from machine $M_1$ by
simply replacing each \vertex\ $u$ with a state $\langle u \rangle$
labeled with $\prop (\langle u \rangle)=\prop(u)$ (recall that by
definition all \vertices\ of $\M_1$ are \nnodes). Thus, for each
edge $(v,w) \in E_1$ we add a transition $(\langle v \rangle,\langle
w \rangle)$ in $\M_1^F$.

For $j>1$, $\M_j^F$ is obtained from $M_j$ by simply replacing each
\boxhsm\ $u$ of $M_j$ with a copy of the Kripke structure
$\M_{\expand(u)}^F$. More precisely, for each \nnode\ $u \in V_j$,
$\langle u\rangle $ is a state of $\M_j^F$ which is labeled with
$\prop(u)$ and for each \boxhsm\ $u\in V_j$ and state $\langle
\alpha \rangle $ of $\M_{\expand(u)}^F$, $\langle u \alpha \rangle$
is a state of $\M_j^F$ and is labeled with $\prop(u) \cup
\prop(\langle  \alpha \rangle)$. The transitions of
$\M_{\expand(u)}^F$ are all inherited in $\M_j^F$, that is, there is
a transition $(\langle u \alpha \rangle,\langle u \beta \rangle)$ in
$\M_j^F$ for each  transition $(\langle  \alpha \rangle,\langle
\beta \rangle)$ of $\M_{\expand(u)}^F$. The remaining transitions of
$\M_j^F$ correspond to the edges of $M_j$:
\begin{itemize}
\item for each \nnode\ $v\in V_j$ and edge $(u,v) \in E_j$
(resp. $((u,z),v) \in E_j$) there
is a transition
from $\langle u \rangle$ (resp. $\langle u z \rangle$) to $\langle v\rangle$;

\item for each \boxhsm\ $v\in V_j$ and
edge $(u,v) \in E_j$ (resp. $((u,z),v) \in E_j$) there
is a transition from $\langle u \rangle$ (resp.  $\langle u z \rangle$)
to $\langle v\, in_{\expand(v)}\rangle$.
\end{itemize}

A \boxhsm\ $u$ expanding into
$M_j$ is a placeholder for  $\M_j^F$ and determines
a subgraph in  $\M^F$ isomorphic to  $\M_j^F$.
This is emphasized in Figure~\ref{peggioDiNoi}, where
we have enclosed in shades of the same shape and color the isomorphic
subgraphs corresponding to a same graph $\M_j^F$.
Therefore, Figure~\ref{peggioDiNoi} also illustrates
the recursive definition of $\M^F$.

If two distinct \boxes\ $u_1$ and $u_2$ both expand into the same
machine $M_j$, that is $\expand(u_1)=\expand(u_2)=h$,
then the states of $\M_j^F$ appear in $\M^F$ in two different scopes,
possibly labeled with different sets of atomic propositions:
in one scope this set contains $\prop(u_1)$ and in the other it
contains $\prop(u_2)$.
The atomic propositions labeling \boxes\ represent
\emph{scope-properties}.
In fact, for a given \boxhsm\ $u$, the set $\prop(u)$ of atomic propositions
is meant to hold true at $u$ and at all its possible descendants.

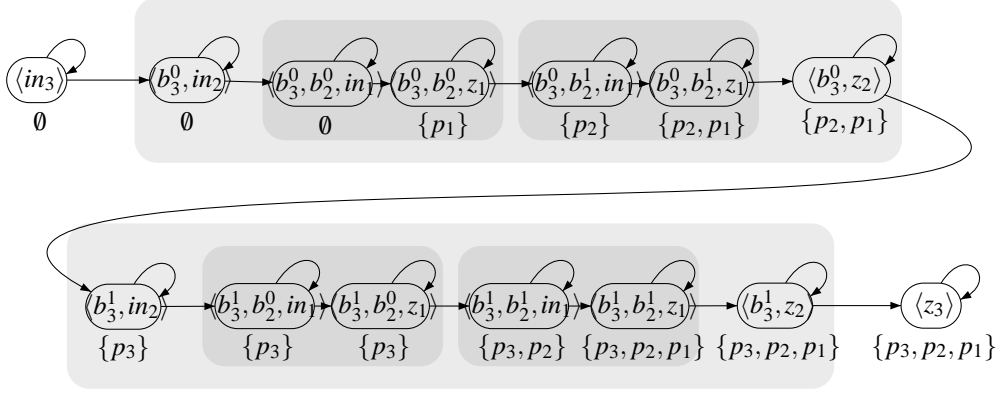
\begin{figure*}[t]
\begin{picture}(213,58)(0,-58)

\drawpolygon[fillgray=0.92,Nframe=n,arcradius=3](18,-1)(18,-23)(120,-23)(120,-1)

\drawpolygon[fillgray=0.85,Nframe=n,arcradius=3](35,-4)(35,-20)(67,-20)(67,-4)

\drawpolygon[fillgray=0.85,Nframe=n,arcradius=3](69,-4)(69,-20)(101,-20)(101,-4)

\drawpolygon[fillgray=0.92,Nframe=n,arcradius=3](9,-31.07)(9,-53.07)(111,-53.07)(111,-31.07)

\drawpolygon[fillgray=0.85,Nframe=n,arcradius=3](27,-34.07)(27,-50.07)(59,-50.07)(59,-34.07)

\drawpolygon[fillgray=0.85,Nframe=n,arcradius=3](61,-34.07)(61,-50.07)(93,-50.07)(93,-34.07)

\node[NLangle=0.0,Nh=6.0,Nmr=3.0](n2)(4.98,-12.0){
$\tuple{\iin_3}$} \nodelabel[NLangle=-90.0,NLdist=5.5](n2){
$\emptyset$}

\node[NLangle=0.0,Nw=13.0,Nh=6.0,Nmr=3.0](n6)(111.98,-12.0){
$\tuple{b_3^0,z_2}$} \nodelabel[NLangle=-90.0,NLdist=5.5](n6){
 $\{p_2,p_1\}$}

\node[NLangle=0.0,Nw=10.0,Nh=6.0,Nmr=3.0](n7)(24.98,-12.0){
$\tuple{b_3^0,\iin_2}$}
\nodelabel[NLangle=-90.0,NLdist=5.5](n7){$\emptyset$}

\node[NLangle=0.0,Nw=13.0,Nh=6.0,Nmr=3.0](n14)(58.57,-12.48){
$\tuple{b_3^0,b_2^0,z_1}$}
\nodelabel[NLangle=-90.0,NLdist=5.5](n14){
 $\{p_1\}$}

\node[NLangle=0.0,Nw=13.2,Nh=6.0,Nmr=3.0](n15)(42.97,-12.48){
$\tuple{b_3^0,b_2^0,\iin_1}$}
\nodelabel[NLangle=-90.0,NLdist=5.5](n15){
 $\emptyset$}

\node[NLangle=0.0,Nw=13.0,Nh=6.0,Nmr=3.0](n57)(92.97,-12.48){
$\tuple{b_3^0,b_2^1,z_1}$}
\nodelabel[NLangle=-90.0,NLdist=5.5](n57){
 $\{p_2,p_1\}$}

\node[NLangle=0.0,Nw=13.2,Nh=6.0,Nmr=3.0](n58)(77.47,-12.48){
$\tuple{b_3^0,b_2^1,\iin_1}$}
\nodelabel[NLangle=-90.0,NLdist=5.5](n58){
 $\{p_2\}$}
\node[NLangle=0.0,Nh=6.0,Nmr=3.0](n63)(123.88,-42.01){
$\tuple{z_3}$} \nodelabel[NLangle=-90.0,NLdist=5.5](n63){
$\{p_3,p_2,p_1\}$}

\node[NLangle=0.0,Nw=10.0,Nh=6.0,Nmr=3.0](n64)(103.11,-42.01){
$\tuple{b_3^1,z_2}$} \nodelabel[NLangle=-90.0,NLdist=5.5](n64){
$\{p_3,p_2,p_1\}$}

\node[NLangle=0.0,Nw=10.0,Nh=6.0,Nmr=3.0](n65)(16.49,-42.15){
$\tuple{b_3^1,\iin_2}$}
\nodelabel[NLangle=-90.0,NLdist=5.5](n65){$\{p_3\}$}

\node[NLangle=0.0,Nw=13.0,Nh=6.0,Nmr=3.0](n68)(50.68,-42.07){
$\tuple{b_3^1,b_2^0,z_1}$}
\nodelabel[NLangle=-90.0,NLdist=5.5](n68){
 $\{p_3\}$}

\node[NLangle=0.0,Nw=13.2,Nh=6.0,Nmr=3.0](n69)(35.28,-42.07){
$\tuple{b_3^1,b_2^0,\iin_1}$}
\nodelabel[NLangle=-90.0,NLdist=5.5](n69){
 $\{p_3\}$}

\node[NLangle=0.0,Nw=13.0,Nh=6.0,Nmr=3.0](n72)(85.18,-42.07){
$\tuple{b_3^1,b_2^1,z_1}$}
\nodelabel[NLangle=-90.0,NLdist=5.5](n72){
 $\{p_3,p_2,p_1\}$}
\node[NLangle=0.0,Nw=13.2,Nh=6.0,Nmr=3.0](n73)(69.08,-42.07){
$\tuple{b_3^1,b_2^1,\iin_1}$}
\nodelabel[NLangle=-90.0,NLdist=5.5](n73){
 $\{p_3,p_2\}$}

\drawedge(n2,n7){}

\drawedge(n68,n73){}

\drawedge(n64,n63){}

\drawloop[ELpos=15,loopdiam=4.0,loopangle=38.0](n2){ }

\drawloop[ELpos=15,loopdiam=4.0,loopangle=38.0](n7){ }

\drawloop[ELpos=15,loopdiam=4.0,loopangle=38.0](n6){ }

\drawloop[ELpos=15,loopdiam=4.0,loopangle=38.0](n65){ }

\drawloop[ELpos=15,loopdiam=4.0,loopangle=38.0](n64){ }

\drawloop[ELpos=15,loopdiam=4.0,loopangle=38.0](n63){ }

\drawloop[ELpos=15,loopdiam=4.0,loopangle=38.0](n73){ }

\drawloop[ELpos=15,loopdiam=4.0,loopangle=38.0](n15){ }

\drawloop[ELpos=15,loopdiam=4.0,loopangle=38.0](n14){ }

\drawloop[ELpos=15,loopdiam=4.0,loopangle=38.0](n58){ }

\drawloop[ELpos=15,loopdiam=4.0,loopangle=38.0](n57){ }

\drawloop[ELpos=15,loopdiam=4.0,loopangle=38.0](n69){ }

\drawloop[ELpos=15,loopdiam=4.0,loopangle=38.0](n68){ }
\drawloop[loopdiam=4.0,loopangle=38.0](n72){ }

\drawedge(n14,n58){}

\drawedge(n15,n14){}

\drawedge(n58,n57){}

\drawedge(n7,n15){}

\drawedge(n65,n69){}

\drawedge(n73,n72){}

\drawedge(n72,n64){}

\drawedge(n69,n68){}
\drawbpedge(n6,-18,85.88,n65,-203,66.37){}

\drawedge(n57,n6){}
\end{picture}
\caption{The Kripke structure obtained by flattening the \VHSM\ $\M$
of Figure~\ref{esempioCHSM}. }\label{peggioDiNoi}
\end{figure*}

\noindent{\large\bf Succinctness.}
Clearly, any  hierarchical structure, either an \HSM\ or an \VHSM,
is in general more succinct than a
traditional Kripke structure. Scope properties make \VHSM\
possibly even more succinct than \HSM.
In fact, two isomorphic subgraphs of a Kripke structure which differ only on the labeling
of the vertices can be represented in an \VHSM\ by the single machine $M_j$,
while it should be represented by two different machines in an \HSM.
Let us recall two main results from \cite{LNPP08} on the succinctness
of these models,
where a  restricted \VHSM\ ${\cal M}$ is an \VHSM\
where for all vertices $u,v$
such that $u$ is an ancestor of $v$ in $\M$ it holds that
$\prop(u) \cap \prop(v)
= \emptyset$.

\begin{theorem}[\cite{LNPP08}]\label{theo:succintezza1}
Restricted \VHSM s can be exponentially more succinct than
\HSM s and finite state machines.
\end{theorem}

There is an exponential gap also between \CHSM s and \VHSM s as
shown in the following proposition.

\begin{theorem}[\cite{LNPP08}]
\VHSM s can be exponentially more succinct than \CHSM s.
\end{theorem}

Observe that  \HSM s, \CHSM s and \VHSM s can all be
translated to equivalent finite state machines with a single exponential
blow-up. Thus, the two succinctness results do not add up to each other, in the
sense that it is not true that \VHSM s can be double exponentially more
succinct than \HSM s.
\ignore{
From Proposition~\ref{theo:succintezza1} and the fact that any
\CHSM\ is also an \VHSM, we have the following.
\begin{corollary}
\VHSM s can be exponentially more succinct than
\HSM s and finite state machines.
\end{corollary}

}

\section{Model checking Problem}\label{sec:Algo}
The \textbf{\ctl model-checking} is the problem of verifying
whether a Kripke structure $\K$ models a \ctl formula.
For an \VHSM\ $\M$,  the \textbf{\ctl model-checking} is the problem of verifying
whether the flat structure $\M^F$ models a \ctl formula.
It is known that the \ctl model-checking problem can be solved in
linear time  in the size of both the formula and  the machine,
see \cite{CE82}, while it is exponential for both  \HSM\ and \VHSM.
More precisely, the following theorem holds.

\begin{theorem}[\cite{AY01},\cite{LNPP08}]\label{theo:ctlMc}
The \ctl\ model-checking  of an \VHSM\ $\M$ for a formula $\varphi$ can be solved in
$O(|\M| \, 2^{|\varphi|\cdot d + |AP_\varphi|})$ time, where $d$ is the maximum number of exit
nodes of $\M$ and $AP_\varphi$ is the set of atomic proposition occurring in $\varphi$. Moreover, if $\M$ is an \HSM, then it can be solved in
$O(|\M| \cdot 2^{|\varphi|\cdot d})$  time.
\end{theorem}

In this section we extend the result to  model-checking a hierarchical structure against a graded-\ctl\-formula.
We first show an algorithm for graded-\ctl\ model-checking  of an \HSM, and then
we  extend it to model-check \VHSM s.

The aim of the algorithm is to determine, for each node $u$ in a machine
$M_j$ of $\M$ and each subformula $\psi$ of $\varphi$, whether  $u$ \emph{satisfies}  $\psi$ or not.
Anyway, the concept of satisfiability  may be  ambiguous, since  whether
 $u$ satisfies  $\psi$ or not  may depend on the possible different sequences of boxes
which expand in $M_j$. Thus, the algorithm transforms $\M$ in such a way that
either  for every  box sequence $b_1,  \dots, b_m $ it holds that
$(\M^F,\langle b_1  \dots  b_m u \rangle) \models \psi $
(and in this case we say that $u$ satisfies $\psi$), or for every  $b_1,  \dots, b_m $ it holds that
$(\M^F,\langle b_1  \dots  b_m u \rangle) \models \neg \psi $.
This transformation determines multiple copies of each  $M_j$, for $j<h$
(clearly, since there are no nodes expanding in the top-level machine $M_h$,   there
is not such ambiguity for a $u\in M_h$).

The algorithm considers the subformulas $\psi$   of $\varphi$, starting from the innermost subformulas, and,
for each node $u$ in $\M$   sets $u.\psi=TRUE$ if $u$ satisfies $\psi$, modifying possibly the hierarchical
structure.
If $\psi$ is an atomic proposition or it is either  $\neg \theta$ or $\theta_1 \wedge \theta_2$, the algorithm is trivial.
For subformulas with temporal operators and grade $0$, then the algorithm behaves exactly as in \cite{AY01} for
the \ctl\ model-checking.
We now show how it behaves for subformulas of the form $\psi=E^{> k} \theta$, with $k > 0$
and  $\theta \in \{  \X \theta_1, \G \theta_1, \theta_1 \U \theta_2 \} $. By inductive hypothesis, we assume that
the algorithm has already set $u.\theta_i=TRUE$ if $u$ satisfies $\theta_i$, for $i=1,2$.

The algorithm for   $\psi=E^{> k} \X \theta_1 $ is rather simple.
It starts from the nodes of $M_1$ setting
$u.\psi=TRUE$ if $u$ satisfies $\psi$, and  then inductively considers  all the machines.
Let $u$ be a node  of $M_j$.
If  $u \notin OUT_j$, then it
satisfies $\psi$ if there are at least $k+1$ successors in $M_j$ satisfying $\theta_1$.
For an output  node $z \in OUT_j$, whether $z$ satisfies  $\psi$ depends also on the successors of a box expanding in $M_j$.
Multiple copies of $M_j$ are then created, denoted $M_j^g$, where  $g:OUT_j \rightarrow \{0,\ldots,k+1\}$,
which correspond  to the different contexts in which $M_j$ occurs.
The  nodes  of $M_j^g$  are   $u^g$, for a node  $u$ of $M_j$,  and the boxes are $b^g$, for a box $b$
of $M_j$.
The idea is that $ g(z)$ is the number of  $z$-successors,
satisfying $\theta_1$,  of a box expanding in $M_j$
(recall that the edges outgoing from a box $b$ are of the type $((b,z),v)$, and we call such $v$
a $z$-successor of $b$).
Thus,  the algorithm sets $z^g.\psi=TRUE$ if the sum of $g(z)$ and the number of successors in $M_j$ satisfying $\theta_1$,
is greater than $k$.
 Moreover, for each box $b$, the algorithm calculates   the number of  $z$-successors of $b$ satisfying  $\theta_1$.
 The new \HSM\ is then obtained by defining the new expansion of  $b$ in $M_j$:
 $b$ expands in the copy  $M_{\expand(b)}^g$ of $M_{\expand(b)}$ such that
$g(z)$ is  the number of  $z$-successors of $b$ satisfying  $\theta_1$.

Consider  now formulas of the type $\psi=E^{> k} \G \theta_1 $ and let us call $\psi^1=E^{> 0}  \G \theta_1 $.

The algorithm  first determines which nodes of the \HSM\ $\M$ satisfy the \ctl
formula $\psi^1$.  At the end of this step
 $\M$ is modified in such a way that  each node $u$ either  satisfies $\psi^1$ or
satisfies $\neg  \psi^1$. In doing that, the size of $\M$ may double (cf.~\cite{AY01}).
Call $S$ 
the set of the nodes satisfying $\psi^1$.

The algorithm determines, for  each node $u\in S$, whether $u$ satisfies $\psi$
using the following idea. Let a {\em sink-cycle} be a cycle
containing only nodes with out-degree $1$.

\noindent{\bf Claim 1.}
Consider
 the graph induced by the states of $\M^F$ where $\psi^1$  holds.
Then, given a state $s $, $(\M^F, s)
\models \psi$   iff in this graph
either there is a \emph{non-sink-cycle} reachable from $s$, or
there are $k+1$ pairwise distinct finite paths connecting $s$ to
\emph{sink-cycles}.

The algorithm  checks the property of the claim analyzing all the machines $M_j$ of $\M$
starting from the bottom-level machine $M_1$, which  contains no boxes.
For each machine $M_j$, it performs a preliminary step to determine
the set of non-sink-cycles $NSC_j \subseteq S$ of nodes $u \in V_j$ such that   a non sink-cycle
is reachable in $\M_j^F$ from  $\langle  u \rangle $, through nodes of $S$.

Then, in a successive  step, the algorithm detects the other nodes satisfying  $\psi$.
In particular for any detected node $u \in V_j$ and
for any sequence $\alpha$ of boxes (below we show how to remove this dependency from $\alpha$)
the following situation can occur:
 \begin{itemize}
 \item
there is a non-sink cycle reachable in $\M^F$ from a state $\langle \alpha u \rangle $ including only nodes in $S$;
\item
 $k+1$ paths start in $\M^F$
from  $\langle \alpha u \rangle$,  each going through nodes belonging to  $S$, and ending into sink-cycles.
\end{itemize}

Observe that, if the  non-sink cycle is in   $\M^F$, but it is not  in $\M_{j} ^F$, then
 $u \notin  NSC_j$ and thus
 the former case has not been detected by the algorithm in the previous preliminary step.

In order to  get that the above properties do not depend on the choice of
$\alpha$, also in  this case multiple copies of each $M_j$ are created, each for a different
context in which $M_j$ occurs.
Each copy is denoted $M_j^g$ where  $g:OUT_j \rightarrow \{0,\ldots,k+1\}$
is a mapping such that if $z$ does not satisfies $\psi^1$ then $g(z)=0$.
Its nodes and boxes are obtained by renaming  nodes and boxes of
$M_j$, as in the previous case.


Let us now give some details on how the above steps are realized.

The set $NSC_j$, for $j \in \{1,\dots,h\}$,
is computed by visiting a graph
$M'_j$, with the nodes in $V_j \cap S$.
If $j \neq 1$, then $M'_j$ contains also
 the boxes $b$ of $M_j$, such that $in_{\expand(b)} \in S$, and
new vertices $(b,z)$, for $z \in OUT_{\expand(b)}\cap S$ (recall that there are no boxes in $M_1$). The
edges of $M_j$ connecting the boxes and the nodes above are edges also of this graph,
moreover, there is an edge from $b$ to $(b,z)$ if there is a path from $in_{\expand(b)}$ to $z$ in $M_{\expand(b)}$,
constituted of all vertices not belonging to $NSC_{\expand(b)}$.

The algorithm  proceeds inductively,
starting from $M_1$. When  $M_j$ is considered, for $j> 1$,  we assume that
the sets  $NSC_{j'}$ have already been determined, for all $j'< j$, and that, for each $z \in OUT_{j'}$,
 it has also been checked
whether there is a path from $in_{j'}$ to $z$ , constituted of all vertices not belonging to $NSC_{j'}$
(observe that  this property is used to define the edges in  $M'_j$).
Moreover, we assume that, if there is such a path, it has also  been  checked whether  there are  vertices in the path with out-degree
greater than $1$ and
whether $z$ has an out-going edge within  $M'_{j'}$.
The result of this test  is useful to detect the non-sink cycles and thus to
determine the  set $NSC_j$.
In fact,   if  either a node $z \in OUT_{\expand(b)}$
 has an out-going edge or there is a vertex with out-degree
at least $2$ in the path from $in_{\expand(b)}$ to $z$, then
a cycle going through  $(b,z)$ in $M'_j$ determines  a non-sink cycle on the corresponding flat machine.

Once the set  $NSC_j$ has been computed, the algorithm sets $u.\psi=TRUE$ for all
$u \in NSC_j$ and then it performs the successive step considering only the remaining nodes.

For each $j$ and each
mapping  $g:OUT_j \rightarrow \{0,\ldots,k+1\}$, a dag $G_j^g$ is constructed with the
nodes $u \in V_j \cap S$ such that  $ u \notin NSC_j$, the boxes $b$ and the new vertices  $(b,z)$, for $z \in OUT_{\expand(b)}$,
 such that  both
$in_{\expand(b)}$ and $z $ satisfy $\psi^1$ and do not belong to $ NSC_{\expand(b)}$, and with
the exception that the sink cycles are substituted
by a single vertex. The edges in  $G_j^g$  are those of  $M_j$.

\ignore{
If $j>h$, we inductively assume that for  $j'<j$ and
for all the nodes $x$  of $M_{j'}^{g'}$,
for a mapping $g': OUT_{j'} \rightarrow \{0,\ldots,k+1\}$,
$x.\psi=TRUE$ if $x$ satisfies $\psi$ and

for a mapping $g': OUT_{j'} \rightarrow \{0,\ldots,k+1\}$,,
the maximum value $m$ in $\{1,\dots,k+1\}$
such that $in_{j'}^{g'}$ satisfies $\psi^m$ has been determined.}

 The algorithm
labels the vertices of $G_j^g$,  starting from  the leaves,   as follows.

\begin{itemize}
\item
 $z \in OUT_j$ is labeled by $g(z)$,
\item
if $x$ in $G_j^g$ is not a box  and
 has successors $x_1,\dots,x_s$, labeled by $l_1,\dots l_s$, then $x$ is labeled by
 $l=max \{l_1+ \dots +l_s, k+1\}$;
\item
for a box $b$, such that $\expand(b)=j'$, let $g'$ be the mapping  such that  $g'(z)=r$ if $(b,z)$
is labeled by $r$, for $z \in OUT_{j'}$.
If  $in_{j'}$ has been labeled by $i$ in the dag  $G_{j'}^{g'}$  then $b$ is labeled $i$ as well
(observe that  the labeling of $in_{j'}$ in   $G_{j'}^{g'}$ has already been determined,  since $j'<j$).
\end{itemize}

As said above,  new machines  $M_j^g$ have been constructed as copies of  $M_j$, by renaming its nodes and boxes.
Now, for each $u \in V_j$, the algorithm sets $u^g.\psi=TRUE$ if $u$ is labeled by $k+1$ in $G_j^g$.

Finally, the expansion mapping for $M_j^g$ is defined as follows: if $\expand_j(b)=j'$ then $b^g$
now expands into $M_{j'}^{g'}$, where $g'$ is such that  $g'(z)=r$  for $z \in OUT_{j'}$ which has been labeled by $r$ in
$G_{j'}^{g'}$.


Finally, for the case of a subformula $\psi =E^{>k} \theta_1 \U \theta_2$, for $k>0$, the algorithm
behaves in a similar way.
It first determines the nodes of $\M$ which satisfy
$E^{>0} \theta_1 \U \theta_2$  and then it
 determines, for  each node $u\in S$, whether $u$ satisfies $\psi$,
with an   approach suggested by the following claim.

\noindent{\bf Claim 2.}
Consider  the graph induced by the states of $\M^F$ where $E^{>0} \theta_1\U \theta_2$   holds,
and by deleting the edges outgoing from states
where $\theta_1$ does not hold.
Then, given a state $s $, $(\M^F, s)
\models \psi$   iff in this graph
either there is a \emph{non-sink-cycle} reachable from $s$, or
there are $k+1$ pairwise distinct finite paths connecting $s$ to
states where $\theta_2$  holds.

Thus, the main difference with respect to the steps described above, is in the definition of
the graphs $M'_j$ and $G_j^g$ since they now do not have edges  outgoing from states
where $\theta_1$ does not hold, in accordance to the Claim 2.
We will omit further details.

Now we can state the first main result,
where $|\varphi|$ is  the number of the boolean and temporal operators in  $\varphi$,
 $d$ is the maximum number of exit
nodes of $\M$ and $\bar k -2$ is the maximal constant occurring in a graded modalities of $\varphi$.

\begin{theorem}\label{theo:ghsmMc}
The graded-\ctl\ model-checking of an \HSM\ $\M$ can be solved in $O(|\M|
\cdot 2^{ |\varphi|\cdot d \cdot log\bar k })$.
\end{theorem}
\begin{proof}
The algorithm sketched above considers the subformulas $\psi$ of $\varphi$, and,
for each node $u$ in $\M$,   sets $u.\psi=TRUE$ if $u$ satisfies $\psi$.
For $\psi=E^{> k} \theta$, with $k > 0$, and $\theta=  \X \theta_1$, the correctness of the algorithm is rather immediate,
while if either  $\theta=  \G \theta_1$ or $\theta =\theta_1 \U \theta_2$,
the correctness of the  algorithm mainly relies on the given claims.
For sake of brevity, we omit here the proof of the claims.

The crucial point is to prove that the algorithm
detects all the nodes $u$ in a machine $M_j$ such that
 a non-sink cycle is reached from $ \langle b_1 \dots b_m u \rangle $  along a path  including only nodes satisfying $E^{>0} \theta$.
 Let $u$ be a node in $M_j$. If there is a  non-sink cycle  reachable
from $\langle u \rangle $ in  $\M_j^F$,
including only nodes in the set $S$ of nodes satisfying $E^{>0} \theta$, then $u \in  NSC_j$
and  the algorithm sets $u.\psi = TRUE$.
Now suppose that there are boxes $b_1,\dots b_m$ and that a  non-sink cycle is reachable from $\langle b_1,\dots b_m u \rangle $ in
$\M_{j'} ^F$ (again including only nodes in $S$) and suppose also that no non-sink cycles are reachable from
$\langle b_r,\dots b_m u \rangle $, for $r > 1$. This implies that
there is $z_1 \in OUT_{\expand(b_1)}$,
and a  non-sink cycle  reachable from $\langle b_1 z_1 \rangle $ in  $\M_{j'} ^F$,
and
there are  $z_1,\dots,z_m$
 such that, for $i=1,\dots,m$,
\begin{itemize}
\item
$z_i \in OUT_{\expand(b_i)}$
\item
 $\langle z_m \rangle$ is reachable from $\langle u\rangle$, in $\M_{j} ^F$,
\item
$\langle z_i \rangle$ is reachable from $\langle b_{i+1},z_{i+1} \rangle$, in $\M_{\expand(b_{i+1})} ^F$

\end{itemize}
In this case the algorithm
sets $(b_1,z_1) \in NSC_{j'} $.
Moreover, in  the new \HSM\, each $b_i$ will expand in a
copy $M_{\expand(b_i)}^{g_i}$ of $M_{\expand(b_i)}$,
where $g_i$ is such that $g_i(z_i)=k+1$. And thus, called $u^g$ the copy of
$u$ in  in $M_j^g$,  the algorithm sets $u^g.\psi = TRUE$
Similarly, the algorithm detects all the nodes $u$ in $M_j$ such that $k+1$ paths  start from
$ \langle b_1 \dots b_m u \rangle $  ending
in sink cycles including only nodes in $S$.
To state the complexity of the algorithm, observe that, while processing a subformula
$\psi=E^{> k} \theta$, with $k > 0$
and  $\theta \in \{ \G \theta_1, \theta_1 \U \theta_2 \} $, the algorithm creates several copies of each
machine $M_j$,
 denoted $M_j^g$ where  $g:OUT_j \rightarrow \{0,\ldots,k+1\}$. Thus the size of the
current \HSM\ grows for a factor not exceeding  $\bar k ^d$, where  $d$ is the maximum number of exit
nodes of $\M$ and $\bar k -2$ is the maximal constant occurring in a graded modalities of $\varphi$.
Since, for each  operator in $\varphi$, the time spent by  the algorithm is linear in the size of the current
\HSM, than the overall running time is $O(|\M|
\cdot \bar k^{ |\varphi|\cdot d })=O(|\M| \cdot 2^{ |\varphi|\cdot d \cdot log\bar k })$.
\end{proof}
Let us remark that, although the multiple copies created by the given algorithm can be seen
as a  step towards the flattening of
the input \HSM, the resulting structure is in general much smaller than the corresponding flat Kripke structure.
To solve the graded-\ctl\ model-checking for \VHSM\ we show now how to reduce it to the model-checking
problem for \HSM.
Let  $\M=(M_1,M_2,\ldots,M_h)$ be an  \VHSM\ and let $\varphi$ be a graded-\ctl\ formula.
Let $AP_\varphi$ be the set of atomic propositions that occur
in $\varphi$. The first step of our algorithm consists of constructing
an \HSM\ $\M_\varphi$ such that $\M_\varphi^F$ is isomorphic to $\M^F$.
Let $index:\{1,\ldots,h\}\times 2^{AP_\varphi}\rightarrow \{1,\ldots,h\,2^{|AP_\varphi|}\}$
be a bijection such that
$index(i,P)<index (j,P')$ whenever
$i<j$. Clearly, $index$ maps  $(i,P)$ into a strictly
increasing sequence of consecutive positive integers starting from $1$.
For a machine
$M_i=(V_i,in_i,\OUT_i, \prop_i,\expand_i, E_i)$, $1\le i\le k$ and
$P\subseteq AP_\varphi$,
define $M_i^P$ as the machine $(V_i^P,in_i^P,\OUT_i^P,
\prop_i^P,\expand_i^P, E_i^P)$
where:
\begin{itemize}
\item $V_i^P=\{u^P\,|\, u\in V_i\}$, and $\OUT_i^P=\{u^P\,|\, u\in \OUT_i\}$;
\item $\prop^P_i(u^P)=\prop_i(u)$ if $u$ is a \nnode\ and
      $\prop^P_i(u^P)=\emptyset$, otherwise;
\item $\expand_i^P(u)=0$ if $u$ is a \nnode\ and
      $\expand_i^P(u)=index(\expand_i(u),P\cup \prop_i(u))$, otherwise;
\item $E_i^P=\{(u^P,v^P)\,|\, (u,v)\in E_i\}\cup
             \{((u^P,z^{P\cup \prop_i(u)}), v^p)\,|\, ((u,z),v)\in E_i\}$.
\end{itemize}
Let $h'=h\,2^{|AP_\varphi|}$. We define $\M_\varphi$ be the tuple of
machines $(M'_1,\ldots,M'_{h'})$
such that for $j=1,\ldots,h'$, $M'_j=M^P_i$ where $j=index(i,P)$.
From the definition of $M^P_i$ it is simple to verify that $\M_\varphi$
is an \HSM\ and $|\M_\varphi|$ is $O(|\M|\,2^{|AP_\varphi|})$.
Moreover, $\M_\varphi^F$ and $\M^F$ coincide,
up to a renaming of the states.
Thus, from Theorem~\ref{theo:ghsmMc}, we have the following second main result.
\begin{theorem}\label{theo:gshsmMc}
The graded \ctl\ model checking of an \VHSM\ can be solved in
$O(|\M| \, 2^{|\varphi|\cdot d \cdot log\bar k + |AP_\varphi|})$ time.
\end{theorem}

\section{Conclusions}\label{sec:Conclusions}
In this paper we have proposed the use of graded-CTL specifications to model-check
hierarchical state machines. We think that the added power in the specification formalism can be fruitfully
exploited in the simulation and testing community to get more meaningful test benches to
perform simulation of more and more complex systems.
We have given algorithms for checking classical \HSM s and so-called \VHSM s.
Let us observe that the alternative approach of model-checking the
fully expanded flat structure
has in general a worse performance because of the exponential gap between an \HSM\ and its corresponding
flat structure. In fact the gain in size of the hierarchical model, is in practice
much greater than the extra exponential factor paid, which depends on the size of (the formula for)
the spe\-ci\-fication, usually quite small.
One last consideration is that we have considered only sequential hierarchical finite state
machines (as an abstraction of the DEVS model).
It is a standard approach, when model checking concurrent systems, to first sequentialize
the model of the SUT (possibly on-the-fly) and then check it with model checking algorithms
for sequential models.
Moreover, the cost of considering
parallel and communicating machines would lead to a double exponential blow-up, the so-called
state explosion problem.
\\\noindent{\bf Acknowledgements.} We thank the anonymous referees for their valuable
comments.
\bibliographystyle{plain}

\end{document}